\documentclass[sn-mathphys,Numbered]{style}

\usepackage{graphicx}%
\usepackage{multirow}%
\usepackage{amsmath,amssymb,amsfonts}%
\usepackage{amsthm}%
\usepackage{mathrsfs}%
\usepackage[title]{appendix}%
\usepackage{xcolor}%
\usepackage{textcomp}%
\usepackage{manyfoot}%
\usepackage{booktabs}%
\usepackage{algorithm}%
\usepackage{algorithmicx}%
\usepackage{algpseudocode}%
\usepackage{listings}%
\usepackage{cleveref}
\DeclareMathAlphabet{\pazocal}{OMS}{zplm}{m}{n}

\newtheorem{theorem}{Theorem}
\newtheorem{remark}{Remark}%

\begin{document}

\title[Free-Surface Equatorial Flows with Surface Tension in Spherical Coordinates]{Free-Surface Equatorial Flows with Surface Tension in Spherical Coordinates}


\author[1,2]{\fnm{Andrei} \sur{Stan}}\email{andrei.stan@ubbcluj.ro}

\affil[1]{ \orgname{Department of Mathematics, Babeș-Bolyai
University}, \orgaddress{ \city{Cluj-Napoca}, \postcode{400084}, \country{Romania}}}

\affil[2]{ \orgname{Tiberiu Popoviciu
Institute of Numerical Analysis, Romanian Academy}, \orgaddress{ \city{Cluj-Napoca}, \postcode{400110}, \country{Romania}}}


\abstract{

In this paper, we determine an exact solution to the governing equations in spherical coordinates for an inviscid, incompressible fluid. This solution describes a steady, purely azimuthal equatorial flow with an associated free surface. Using functional analytic techniques, we demonstrate that if a free surface is known beforehand, the variations in pressure needed to achieve this surface implicitly define the shape of the free surface  in a unique way.
}



\maketitle

\section{Introduction}

This paper focuses on developing an exact solution for the governing
equations of geophysical fluid dynamics (GFD) governing inviscid,
incompressible, and stratified fluid dynamics in the equatorial region.
Specifically, we investigate a steady purely azimuthal flow. Our approach
ensures accuracy by consistently employing spherical coordinates, thereby
avoiding any simplifications to the geometry within the governing equations.
Additionally, within these equations, we incorporate the complexities of
Coriolis and centripetal forces, along with the surface tension resulting in highly nonlinear dynamics \cite%
{const2016a, cushman2011introduction, vallis2006atmospheric}.

This study builds upon recent significant advancements in the field. In \cite%
{const2016a, const2016b}, the authors shown for the first time that exact
solutions to the complete governing equations of geophysical fluid dynamics
(GFD) can be formulated using spherical coordinates, representing purely
azimuthal, depth-varying flows. These solutions are capable of modeling both
the equatorial undercurrent (EUC) and the Antarctic Circumpolar Current
(ACC). Initially, these solutions describe purely homogeneous fluids without
stratification. Subsequently, in works by the authors \cite{henry2018a,
henry2018b}, exact equatorial flow solutions allowing for stratification
were developed. These solutions, while relatively simple, feature fluid
density varying linearly with depth and remaining independent of latitude.
Finally, in \cite{M}, the authors established the existence of solutions in
the context of general stratification, a paper that the present work aims to
build upon and enhance.

In previous works (see, e.g., \cite{as, M, petrusel}), the existence and uniqueness of the free surface are guaranteed only for pressure variations close to the pressure required to achieve a flat, undisturbed surface. The novelty of this paper is that if a free surface $\gamma$ is known beforehand, a similar conclusion holds: sufficiently small variations in pressure from the value required to obtain $\gamma$ lead to a unique free surface. From a physical point of view, such a result is expected, as small variations in pressure that determine a smooth enough free surface should also result in a similarly smooth free surface.
\section{Preliminaries}


Throughout this paper we employ spherical coordinates $(r, \theta, \phi)$,
where $r$ represents distance to the sphere's center, $\theta$ ranges from 0
to $\pi$ (with $\frac{\pi}{2} - \theta$ denoting latitude), and $\phi$ spans
from 0 to $2\pi$ (longitude). Notably, the North and South poles correspond
to $\theta = 0$ and $\pi$, respectively, while the Equator lies at $\theta = 
\frac{\pi}{2}$. Unit vectors in this system, denoted as $\mathbf{e}_r$, $%
\mathbf{e}_{\theta}$, and $\mathbf{e}_{\phi}$, provide directional
references, with $\mathbf{e}_{\phi}$ from West to East and $\mathbf{e}%
_{\theta}$ from North to South.

The equations describing the motion of an inviscid and incompressible flow
are comprised in the Euler equations and, respectively, mass conservation
equation. In spherical coordinates, they are given by 
\begin{align}  \label{euler eq}
\begin{aligned} u_t + uu_r + \frac{v}{r}u_{\theta} +
\frac{w}{r\sin\theta}u_{\phi} - \frac{1}{r}(v^2 + w^2) &= -\frac{1}{\rho}
p_r + F_r \\ v_t + uv_r+ \frac{v}{r}v_{\theta} +
\frac{w}{r\sin\theta}v_{\phi} + \frac{1}{r}(uv - w^2\cos\theta) &=
-\frac{1}{\rho} \frac{1}{r}p_\theta + F_{\theta} \\ w_t + uw_r +
\frac{v}{r}w_{\theta} + \frac{w}{r\sin\theta}w_{\phi} + \frac{1}{r}(uw +
vw\cot\theta) &= -\frac{1}{\rho} \frac{1}{r\sin\theta}p_\phi + F_{\phi},
\end{aligned}
\end{align}
and 
\begin{equation}  \label{mass conservation}
\frac{1}{r^2}\frac{\partial}{\partial r}(\rho r^2u) + \frac{1}{r\sin\theta}%
\frac{\partial}{\partial\theta}(\rho v \sin \theta) + \frac{1}{r\sin\theta}%
\frac{\partial(\rho w)}{\partial\phi} = 0.
\end{equation}
Here, $\mathbf{u} = u\mathbf{e}_r + v\mathbf{e}_{\theta} + w\mathbf{e}_{\phi}
$ is the velocity field, $p(r,\theta,\phi)$ represents the pressure, $(F_r,
F_\theta, F_\phi)$ the body-force vector, while $\rho=\rho(r,\theta)$ stands
for the density distribution.

To ensure an accurate analysis of fluid dynamics in specific locations, it
is crucial to consider the effects of Earth's rotation. Additional terms
must be integrated into the Euler equations to account for these effects,
namely the Coriolis force $2\boldsymbol{\Omega} \times \mathbf{u}$ and the
centripetal acceleration $\boldsymbol{\Omega} \times (\boldsymbol{\Omega}
\times \mathbf{r})$, where 
\begin{equation*}
\boldsymbol{\Omega} = \Omega(\mathbf{e}_r \cos\theta - \mathbf{e}_{\theta}
\sin\theta), 
\end{equation*}

Here, $\Omega \approx 7.29 \times 10^{-5}$ rad/s represents the Earth's
constant rate of rotation and $\mathbf{r}=r \mathbf{e}_r$. These forces are
incorporated into the Euler equations. Consequently, with gravity given by
the vector $(-g, 0, 0)$, and the latter two quantities combined, 
\begin{align*}
2\boldsymbol{\Omega} \times \mathbf{u}+\boldsymbol{\Omega} \times (%
\boldsymbol{\Omega} \times \mathbf{r})&=2\Omega \left( -w\sin\theta\mathbf{e}%
_r - w\cos\theta\mathbf{e}_{\theta} + (u\sin\theta + v\cos\theta)\mathbf{e}%
_{\phi}\right) \\
& -r\Omega^2 \left( \sin^2\theta \mathbf{e}_r+\sin \theta \cos \theta 
\mathbf{e}_\phi \right),
\end{align*}
the Euler's equations become 
\begin{align}  \label{euler eq2}
\begin{aligned} uu_r + \frac{v}{r}u_{\theta} + \frac{w}{r\sin\theta}u_{\phi}
- \frac{1}{r}(v^2 + w^2) -2\Omega w \sin\theta - r\Omega^2 \sin^2\theta &=
-\frac{1}{\rho} p_r -g \\ uv_r+ \frac{v}{r}v_{\theta} +
\frac{w}{r\sin\theta}v_{\phi} + \frac{1}{r}(uv - w^2\cos\theta)-2\Omega w
\cos\theta - r\Omega^2 \sin\theta \cos\theta &= -\frac{1}{\rho}
\frac{1}{r}p_\theta \\ uw_r + \frac{v}{r}w_{\theta} +
\frac{w}{r\sin\theta}w_{\phi} + \frac{1}{r}(uw + vw\cot\theta) +2\Omega
\left(u \sin\theta + v \cos\theta \right)&= - \frac{1}{\rho
r\sin\theta}p_\phi, \end{aligned}
\end{align}
where we considered $u_t=v_t=w_t=0$. In addition to the mass conservation
and Euler equations, the movement of water is influenced by boundary
conditions. As the fluid extends infinitely in all horizontal directions, it
encounters two boundaries: the rigid flat bed and the water's free surface.
On the bed where $r = d(\theta, \phi)$, we apply the kinematic boundary
condition, 
\begin{equation*}
u = \frac{v}{r} d_\theta + \frac{w}{r \sin \theta} d_\phi, 
\end{equation*}
while on the free surface where $r = R + h(\theta, z)$, we impose the
kinematic boundary condition 
\begin{equation*}
u = \frac{v}{r} h_\theta + \frac{w}{r \sin \theta} h_\phi. 
\end{equation*}
Additionally, on the free surface, we have the dynamic condition 
\begin{equation}  \label{dinamic condition 1}
p = P(\theta, \phi) + \sigma \nabla \cdot \vec{n}.
\end{equation}%
Here, $\sigma$ represents the coefficient of surface tension, and $\vec{n}$
is the outward-pointing unit normal vector. The function $h$ is currently
unknown and will be determined later.

Next, we recall two well known results from literature used in this paper.
We start with the Implicit Function Theorem, which ensures the existence of
nontrivial zeros of a $C^1$ mapping between Banach spaces (see, e.g., \cite%
{functionala}).

\begin{theorem}
\label{teorema functiei implicie}  Let $X,Y,Z$ be Banach spaces, $U\subset
X\times Y$ an open neighbourhood of a point $(x_0,y_0)\in X\times Y$ and let 
$f\colon U \to Z$ be a continuous functions. Assume that:

\begin{itemize}
\item[i) ] The function $f$ satisfies $f(x_0,y_0)=0$. 

\item[ii)] The partial derivative $f_y(x_0,y_0)$ exists and is an linear
homeomorphism from $Y$ to $Z$. 
\end{itemize}

Then, there exists an open neighbourhood $U_1$ of $x_0$ and a unique $%
g\colon U_1 \to Y$ continuous function such that $g(x_0)=y_0$ and $%
f(x,g(x))=0$ on $U_1$.
\end{theorem}

The second result concerns the existence and uniqueness of solution for a
second order differential equation (see, e.g., \cite{logan, Coddington}).

\begin{theorem}
\label{th initial value pb}  The initial value problem 
\begin{equation}  \label{ec dif generala}
y^{\prime \prime }+ p(t)y^{\prime }+ q(t)y = g(t), \quad y(t_0) = y^{\prime
}(t_0) = 0,
\end{equation}
where $p, q, $ and $g $ are continuous functions on an open interval $I $
that contains the point $t_0 $, has a unique solution. In addition, the
solution $y(t)$ is given by  
\begin{equation*}
y(t)= y_p+\int_{t_0}^{t}\frac{ g(s) \left(\varPhi_1(t) W_1(s) +\varPhi_2(t)
W_2(s)\right)}{W(\varPhi_1, \varPhi_2)(s)} ds,
\end{equation*}
where $y_p$ is a particular solution of \eqref{ec dif generala}, $\{\varPhi%
_1, \varPhi_2\}$ is a basis for the solutions of the homogeneous equation $%
y^{\prime \prime }+ p(t)y^{\prime }+ q(t)y =0$, $W$ is the Wronskian of the
basis $\{\varPhi_1, \varPhi_2\}$ and $W_i$ (i=1,2) is the Wronskian obtained
by replacing the $i-$th column of $W$ with the column vector $(0,1)$.

\end{theorem}

\section{Main result}

In the sequel, we assume an azimuthal flow, i.e., the velocity field $%
\mathbf{u} = u\mathbf{e}_r + v\mathbf{e}_{\theta} + w\mathbf{e}_{\phi}$
satisfies $u=v=0$ and $w=w(r, \theta)$. Also, we assume the free surface satisfy $h=h(\theta)$  while the flat bed  is a fixed real value $d$. Consequently, the mass conservation equation \eqref{mass
conservation} is automatically satisfied, while the Euler's equations %
\eqref{euler eq2} takes the form 
\begin{equation}  \label{ec euler3}
\begin{cases}
-\frac{w^2}{r} - 2\Omega w \sin\theta - r\Omega^2 \sin^2\theta = -\frac{1}{%
\rho}p_r - g, \\ 
-\frac{w^2}{r} \cot\theta - 2\Omega w \cos\theta - r\Omega^2 \sin\theta
\cos\theta = -\frac{1}{\rho r} p_\theta, \\ 
0 = -\frac{1}{\rho} \frac{1}{r \sin\theta} p_\phi,%
\end{cases}%
\end{equation}
or equivalently 
\begin{equation}
\begin{cases}
\frac{\rho(r,\theta)}{r}\left( w+\Omega r \sin \theta\right)^2=p_r+g \rho(r,
\theta) \\ 
\rho(r, \theta) \cot \theta \left( w+\Omega r \sin \theta\right)^2=p_\theta
\\ 
0=p_\phi.%
\end{cases}%
\end{equation}
From the last relation, we obtain the independence of the pressure with
respect to $\phi$, i.e., $p=p(r, \theta)$. If we denote 
\begin{equation*}
Z = Z(r,\theta) := \frac{1}{r}\left( w + \Omega r \sin \theta \right)^2,
\end{equation*}
and differentiate the first equation with respect to $\theta$ and the second
one with respect to $r$, we deduce 
\begin{equation*}
(\rho Z)_{\theta} - (\rho r \cot \theta Z)_r = g \rho \theta,
\end{equation*}
which yields 
\begin{equation}  \label{ec U}
-(r \cos \theta) U_r + (\sin \theta) U_{\theta} = (r \sin \theta) g
\rho_\theta (r,\theta),
\end{equation}
where $U(r, \theta)=\rho(r, \theta)\left( w + \Omega r \sin \theta \right)^2.
$

Employing the method of characteristics and following steps similar to those in 
\cite{M,petrusel}, we infer  
\begin{equation}
w(r,\theta) = -\Omega r \sin \theta + \frac{1}{\sqrt{\rho(r,\theta)}}\left(
F(r \sin \theta) +gr \sin \theta \int_0^{f(\theta)} \rho_\theta ( \bar{r}%
(s), \bar{\theta}(s))ds \right)^{\frac{1}{2}},  \label{w}
\end{equation}
where 
\begin{equation*}
\bar{r}(s)=\frac{r\sin \theta}{2}(e^s+e^{-s}),\, \bar{\theta}(s)=\arccos\left(%
\frac{1-e^{2s}}{1+e^{2s}} \right), \, f(\theta)=\frac{1}{2}\ln \frac{1-\cos
\theta}{1+\cos \theta},
\end{equation*}
and $F(\xi):=U\left(\xi, \frac{\pi}{2}\right)$ is an arbitrary smooth
function. Note that \begin{equation*}
    \bar{\theta}(f(\theta))=\theta \text{ and }\bar{r}(f(\theta))=r.
\end{equation*} To proceed with the determination of the pressure, one sees that 
\begin{align}  \label{p_r}
p_r+g\rho(r, \theta)&=\frac{\rho(r, \theta)}{r}\left( w+\Omega r \sin
\theta\right)^2 \\
&=\frac{1}{r}\left( F(r \sin \theta) +gr \sin \theta \int_0^{f(\theta)}
\rho_\theta ( \bar{r}(s), \bar{\theta}(s))ds \right),  \notag
\end{align}
and 
\begin{equation}  \label{p_tetha}
p_\theta=\cot \theta U(r\theta)=\cot \theta \left( F(r \sin \theta) +gr \sin
\theta \int_0^{f(\theta)} \rho_\theta ( \bar{r}(s), \bar{\theta}(s))ds
\right).
\end{equation}
After integration from $a$ (an arbitrary constant) to $r$ in \eqref{p_r}, we
infer 
\begin{equation}  \label{p1}
p(r,\theta) = C(\theta) - g \int_{r}^{a} \rho(s,\theta) \, ds + \int_{a \sin
\theta}^{r \sin \theta} \frac{F(s)}{s} \, ds + \pazocal{F}(s,\theta) \, dy,
\end{equation}
where $C(\theta):=p(a,\theta)$ is an arbitrary smooth function and 
\begin{equation*}
\pazocal{F}(\xi,\theta)= \int_{0}^{f(\theta)} g \rho_{\theta} \left( \xi\, 
\frac{e^s + e^{-s}}{2}, \bar{\theta}(s) \right) \, ds.
\end{equation*}
To determine the function $C$, we differentiate with respect to $\theta$ in %
\eqref{p1}, which gives 
\begin{align*}
p_\theta&=C^{\prime }(\theta)-g\int_0^r \rho_\theta(s,\theta)ds+\cot \theta
\left( F(r \sin \theta)-F(a\sin \theta) \right) \\
& \quad +\pazocal{F} (r\sin \theta) r \cos \theta-\pazocal{F}(a \sin \theta)
a \cos \theta+\int_{a\sin \theta}^{r \sin \theta}\pazocal{F}_\theta(s,
\theta)ds.
\end{align*}
Since 
\begin{align*}
\pazocal{F}_\theta(s, \theta)&=f'(\theta) g \rho_\theta\left( 
\frac{\bar{s}\left( f(\theta)\right)}{\sin \theta}, \bar{\theta}(f(\theta))
\right)\\&
=\frac{g}{\sin \theta} \rho_\theta \left( \frac{%
s}{\sin \theta},\theta \right),
\end{align*}
we obtain 
\begin{equation*}
\int_{a\sin \theta}^{r \sin \theta}\pazocal{F}_\theta(s, \theta)ds=\int_a^r
g \rho_\theta (s, \theta) ds.
\end{equation*}
Hence, the expression of $p_\theta$ becomes 
\begin{equation}  \label{p_theta2}
p_{\theta} = C^{\prime }(\theta) + \cot \theta [F(r \sin \theta) - F(a \sin
\theta)] + \pazocal{F}(r \sin \theta) r \cos \theta-\pazocal{F}(a\sin
\theta) a \cos\theta.
\end{equation}
Consequently, from \eqref{p_tetha} and \eqref{p_theta2}, we deduce 
\begin{equation}
C^{\prime }(\theta) = F(a \sin \theta) \cot \theta + F(a \sin \theta) a \cos
\theta.  \label{Cprim}
\end{equation}
Finally, relations \eqref{p1} and \eqref{Cprim} yields, 
\begin{align}  \label{formula presiune}
p(r, \theta)&=b-g\int_a^r \rho(s, \theta) ds+\int_{a\sin \theta}^{r \sin
\theta}\left(\frac{F(s)}{s}+\pazocal{F}(s, \theta) ds\right) \\
& \quad +\int_{\tfrac{\pi}{2}}^ \theta \left( F(a\sin s)\cot s+\pazocal{F}(a
\sin s)a \cos s \right)ds,
\end{align}
where $a,b$ are arbitrary real numbers.

\subsection{The dynamic condition}

In the sequel, we analyze the dynamic boundary condition \eqref{dinamic
condition 1} on the free surface $r=R+h( \theta)$. Since we assumed an
azimuthal flow and $h$ is independent of the depth, relation \eqref{dinamic condition 1} takes
the form 
\begin{equation}  \label{dinamic condition 2}
p(R+h(\theta),\theta)=P(\theta)+\sigma \nabla \cdot \vec{n},
\end{equation}
where $P$ is the known pressure on the surface. Since the implicit equation
of the free surface is $H(r, \theta, z) := r - R - h(\theta)=0$ and the
gradient in spherical coordinates has the representation (see, e.g., \cite[%
Chapter 1]{griffin}),  
\begin{equation*}
\nabla=\mathbf{e}_r\partial_r+\mathbf{e}_\theta\frac{1}{r}\partial_\theta+%
\mathbf{e}_\phi\frac{1}{r \sin \theta}\partial _\phi,
\end{equation*}
we obtain the normal derivative to the surface $H$ as 
\begin{equation*}
\vec{N}=\mathbf{e}_r-\displaystyle\frac{h_\theta}{r}\mathbf{e}_\theta.
\end{equation*}
Therefore, the pointing unit normal vector is 
\begin{equation*}
\vec{n}=\displaystyle\frac{\vec{N}}{\lVert \vec{N} \rVert}=\displaystyle%
\frac{r}{\sqrt{r^2+{h_\theta}^2}}\cdot \vec{N}=\displaystyle\frac{r}{\sqrt{%
r^2+{h_\theta}^2}} \mathbf{e}_r-\displaystyle\frac{h_\theta}{\sqrt{r^2+{%
h_\theta}^2}} \mathbf{e}_\theta.
\end{equation*}
Whence, given the the divergence in spherical coordinates (see, e.g., \cite[%
Chapter 1]{griffin}), 
\begin{equation*}
\nabla \cdot w= \frac{1}{r^2}\partial_r(r^2 w_r)+\frac{1}{r \sin \theta}
\partial_\theta (\sin \theta \, w_\theta)+\frac{1}{r \sin \theta}%
\partial_\phi w_\phi,
\end{equation*}
one has, 
\begin{align}  \label{div n}
\nabla\cdot \vec{n}&=\displaystyle\frac{1}{r^2}\partial _r(r^2 n_r)+%
\displaystyle\frac{1} {r \sin \theta}\partial_\theta(n_\theta \,\sin \theta )
\\
& = \frac{2r^2+3h_\theta^2}{(r^2+h_\theta^2)^{\frac{3}{2}}}- \frac{%
rh_{\theta \theta}}{(r^2+h_\theta^2)^{\frac{3}{2}}}-\frac{\cot \theta\,
h_\theta }{r (r^2+h_\theta^2)^{\frac{1}{2}}}  \notag \\
& = \frac{2(R+h(\theta))^2+3h_\theta^2}{((R+h(\theta))^2+h_\theta^2)^{\frac{3%
}{2}}}-\frac{(R+h(\theta))h_{\theta \theta}}{((R+h(\theta))^2+h_\theta^2)^{%
\frac{3}{2}}}  \notag \\
& \quad -\frac{\cot \theta\, h_\theta }{(R+h(\theta))
((R+h(\theta))^2+h_\theta^2)^{\frac{1}{2}}}.  \notag
\end{align}
Finally, from \eqref{dinamic condition 2} and \eqref{div n}, we find the
pressure on the free surface in the form of a Bernoulli type problem 
\begin{align}  \label{bernoulli}
P(\theta)&=b-g\int_a^{R+h(\theta)} \rho(s, \theta) ds+\int_{a\sin
\theta}^{(R+h(\theta)) \sin \theta}\left(\frac{F(s)}{s}+\pazocal{F}(s,
\theta) ds\right) \\
& \quad +\int_{\tfrac{\pi}{2}}^ \theta \left( F(a\sin s)\cot s+\pazocal{F}(a
\sin s)a \cos s \right)ds  \notag \\
&\quad -\sigma \frac{2(R+h(\theta))^2+3h_\theta^2}{((R+h(\theta))^2+h_%
\theta^2)^\frac{3}{2}}+\sigma \frac{h_{\theta\theta}(R+h(\theta))}{\nonumber{%
((R+h(\theta))^2+h_\theta^2)^\frac{3}{2}}} \\
& \quad + \sigma \frac{\cot \theta\, h_\theta }{(R+h(\theta))
((R+h(\theta))^2+h_\theta^2)^{\frac{1}{2}}} .  \notag
\end{align}

\subsection{The existence of solutions for the free surface}

To have a meaningful comparison of the physical quantities, we start by 
non-dimensionalizing equation \eqref{bernoulli}. For this, we denote with $
P_{atm}$ the pressure corresponding to an undisturbed free surface, i.e., 
the pressure $P$ obtained by setting $h\equiv 0$ in \eqref{bernoulli}. 
Subsequently, we denote  
\begin{equation}
\mathcal{\wp }=\frac{P}{P_{atm}}\text{ and }\mathtt{h}=\frac{h}{R}.
\label{notatii1}
\end{equation}
Whence, relation \eqref{bernoulli} together with \eqref{notatii1}, leads us 
to the functional equation  
\begin{equation}
\mathbb{F}(\mathtt{h},\mathcal{\wp })=0,  \label{ec abstracta}
\end{equation}
where  
\begin{align}
\mathbb{F}(\mathtt{h},\mathcal{\wp })& =-\mathcal{\wp }\left( \theta \right)
-\frac{g}{P_{atm}}\int_{a}^{R(1+\mathtt{h}\left( \theta \right) )}\rho
\left( s,\theta \right) ds  \label{ec principala} \\
& \quad +\frac{1}{P_{atm}}\int_{a\sin \theta }^{R(1+\mathtt{h}\left( \theta
\right) )\sin \theta }\left( \frac{\nonumber F\left( s\right) }{s}+\mathcal{%
F }\left( s,\theta \right) \right) ds \\
& \quad -\frac{\sigma }{RP_{atm}}\left( \frac{(2\,(1+\mathtt{h})^{2}+3\, 
\mathtt{h}_{\theta }^{2}}{((1+\mathtt{h})^{2}+\mathtt{h}_{\theta }^{2})^{ 
\frac{3}{2}}}-\frac{\mathtt{h}_{\theta \theta }(1+\mathtt{h})}{((1+\mathtt{h}
)^{2}+\mathtt{h}_{\theta }^{2})^{\frac{3}{2}}}-\frac{\cot \theta \,\mathtt{h}
_{\theta }}{(1+\mathtt{h})\left( (1+\mathtt{h})^{2}+\mathtt{h}_{\theta
}^{2}\right) ^{\frac{1}{2}}}\right)  \notag \\
& \quad +\varphi \left( \theta \right) ,  \notag
\end{align}
and  
\begin{equation*}
\varphi \left( \theta \right) =\frac{b}{P_{atm}}+\frac{1}{P_{atm}}\int_{\pi
/2}^{\theta }\left( F\left( a\sin s\right) \cot s+\mathcal{F}\left( a\sin
s\right) a\cos s\right) ds.
\end{equation*}

Our goal is to establish a relationship between the pressure variations at
the free surface and the shape of the surface, which corresponds to finding
a nontrivial solution to the functional equation \eqref{ec abstracta}. To
achieve this, we assume the existence of a smooth enough "trivial" free surface, denoted
by $\bar{\gamma}$, which is derived from specific measurements. 
The non-dimensionalized pressure   $\mathcal{\wp}_{\gamma}$ required to maintain this shape is obtained by
substituting $\mathtt{h}$ with $\gamma = \frac{\bar{\gamma}}{R}$ in %
\eqref{ec principala}. 

\begin{remark}
If $\bar{\gamma} = 0$, then $\mathcal{\wp}_{\gamma}$ represents the pressure
required to maintain a flat, undisturbed free surface. 
\end{remark}

Let 
\begin{equation*}
C_{\gamma }=\left\{ u\in C^{2}\left[ \tfrac{\pi }{2}-\varepsilon ,\tfrac{\pi 
}{2}+\varepsilon \right] :u\left( \tfrac{\pi }{2}\right) =\gamma \left( 
\tfrac{\pi }{2}\right) ,\ u^{\prime }\left( \tfrac{\pi }{2}\right) =\gamma
^{\prime }\left( \tfrac{\pi }{2}\right) \right\} ,
\end{equation*}%
where $\varepsilon =0.016$ corresponds to a strip approximately 100 km wide
around the Equator \cite{cj}. Clearly, $\mathbb{F}$ defines a continuously
differentiable map 
\begin{equation*}
\mathbb{F}:C_{\gamma }\times C\left[ \tfrac{\pi }{2}-\varepsilon ,\tfrac{\pi 
}{2}+\varepsilon \right] \rightarrow C\left[ \tfrac{\pi }{2}-\varepsilon ,%
\tfrac{\pi }{2}+\varepsilon \right] .
\end{equation*}%
In the subsequent, we analyze the Fr\'{e}chet derivative of $\mathbb{F}$
with respect to the first argument. To this aim, we have 
\begin{align*}
\mathbb{F}(t\mathtt{h}+\gamma ,\mathcal{\wp }_{\gamma })-\mathbb{F}(\gamma ,%
\mathcal{\wp }_{\gamma })& =-\frac{g}{P_{atm}}\int_{R\left( 1+\gamma \right)
}^{R\left( 1+t\mathtt{h}+\gamma \right) }\rho (s,\theta )ds \\
& \quad +\frac{1}{P_{atm}}\int_{R\left( 1+\gamma \right) \sin \theta }^{R(1+t%
\mathtt{h}+\gamma )\sin \theta }\left( \frac{F\left( s\right) }{s}+%
\pazocal{F}\left( s,\theta \right) \right) ds \\
& \quad -\left( J\left( t\mathtt{h}+\gamma \right) -J\left( \gamma \right)
\right) ,
\end{align*}%
where 
\begin{equation*}
J(v):=\frac{\sigma }{RP_{atm}}\left( \frac{2(1+v)^{2}+3v_{\theta }^{2}}{%
((1+v)^{2}+v_{\theta }^{2})^{\frac{3}{2}}}-\frac{v_{\theta \theta }(1+v)}{%
((1+v)^{2}+v_{\theta }^{2})^{\frac{3}{2}}}-\frac{\cot \theta \,v_{\theta }}{%
(1+v)((1+v)^{2}+v_{\theta }^{2})^{\frac{1}{2}}}\right) .
\end{equation*}%
Using the mean value theorem, we easily derive 
\begin{align*}
& \lim_{t\rightarrow 0}\frac{1}{t}\left( -\frac{g}{P_{atm}}\int_{R\left(
1+\gamma \right) }^{R\left( 1+t\mathtt{h}+\gamma \right) }\rho (s,\theta
)ds\right) +\frac{1}{P_{atm}}\int_{R\left( 1+\gamma \right) \sin \theta
}^{R(1+t\mathtt{h}+\gamma )\sin \theta }\left( \frac{F\left( s\right) }{s}+%
\pazocal{F}\left( s,\theta \right) \right) ds \\
& =\psi (\theta )\mathtt{h},
\end{align*}%
where 
\begin{align*}
\psi (\theta )& =-\frac{Rg}{P_{atm}}\rho \left( R\left( 1+\gamma \right)
,\theta \right) +\frac{R\sin \theta }{P_{atm}}\left( \frac{F\left( R\left(
1+\gamma \right) \sin \theta \right) }{R\left( 1+\gamma \right) \sin \theta }%
+\pazocal{F}\left( R\left( 1+\gamma \right) \sin \theta ,\theta \right)
\right)  \\
& =\frac{\rho (R(1+\gamma ),\theta )}{P_{atm}}\left( -gR+(w(R(1+\gamma
),\theta )+\Omega R(1+\gamma )\sin \theta )^{2}\right) .
\end{align*}%
To compute the derivative of $J$ in $\gamma$, we denote
\begin{eqnarray*}
a\left( t\right)  &=&\frac{2(1+t\mathtt{h}+\gamma )^{2}+3\left( t\mathtt{h}%
_{\theta }+\gamma _{\theta }\right) ^{2}}{\left( (1+t\mathtt{h}+\gamma
)^{2}+\left( t\mathtt{h}_{\theta }+\gamma _{\theta }\right) ^{2}\right) ^{%
\frac{3}{2}}}, \\
b\left( t\right)  &=&\frac{\left( t\mathtt{h}_{\theta \theta }+\gamma
_{\theta \theta }\right) (1+t\mathtt{h}+\gamma )}{\left( (1+t\mathtt{h}%
+\gamma )^{2}+\left( t\mathtt{h}_{\theta }+\gamma _{\theta }\right)
^{2}\right) ^{\frac{3}{2}}}, \\
c\left( t\right)  &=&\frac{\cot \theta \left( \,t\mathtt{h}_{\theta }+\gamma
_{\theta }\right) }{(1+t\mathtt{h}+\gamma )\left( (1+t\mathtt{h}+\gamma
)^{2}+\left( t\mathtt{h}_{\theta }+\gamma _{\theta }\right) ^{2}\right) ^{%
\frac{1}{2}}}.
\end{eqnarray*}%
Since $a,b,c$ are smooth functions around $0,$ one has 
\begin{equation*}
\lim_{t\rightarrow 0}\frac{1}{t}\left( J\left( t\mathtt{h}+\gamma \right)
-J\left( \gamma \right) \right) =\frac{\sigma }{RP_{atm}}\left( a^{\prime
}\left( 0\right) -b^{\prime }\left( 0\right) -c^{\prime }\left( 0\right)
\right) .
\end{equation*}%
Simple computations yield 
\begin{eqnarray*}
a^{\prime }\left( 0\right)  &=&\frac{4\mathtt{h}(1+\gamma )+6\mathtt{h}%
_{\theta }\gamma _{\theta }}{\left( (1+\gamma )^{2}+\gamma _{\theta
}^{2}\right) ^{\frac{3}{2}}}-3\left( 2\left( 1+\gamma \right) ^{2}+3\gamma
_{\theta }^{2}\right) \frac{\left( \mathtt{h}\left( 1+\gamma \right) +%
\mathtt{h}_{\theta }\gamma _{\theta }\right) }{\left( (1+\gamma )^{2}+\gamma
_{\theta }^{2}\right) ^{\frac{5}{2}}}, \\
b^{\prime }\left( 0\right)  &=&\frac{\mathtt{h}_{\theta \theta }(1+\gamma
)+\gamma _{\theta \theta }\mathtt{h}}{\left( (1+\gamma )^{2}+\gamma _{\theta
}^{2}\right) ^{\frac{3}{2}}}-3\gamma _{\theta \theta }(1+\gamma )\frac{%
\left( \mathtt{h}\left( 1+\gamma \right) +\mathtt{h}_{\theta }\gamma
_{\theta }\right) }{\left( (1+\gamma )^{2}+\gamma _{\theta }^{2}\right) ^{%
\frac{5}{2}}}, \\
c^{\prime }\left( 0\right)  &=&\cot \theta \frac{\mathtt{h}_{\theta }}{%
\left( 1+\gamma \right) \left( \left( 1+\gamma \right) ^{2}+\gamma _{\theta
}^{2}\right) ^{\frac{1}{2}}}-\cot \theta \frac{\mathtt{h}\gamma _{\theta }}{%
\left( 1+\gamma \right) ^{2}\left( \left( 1+\gamma \right) ^{2}+\gamma
_{\theta }^{2}\right) ^{\frac{1}{2}}} \\
&&-\cot \theta \frac{2\gamma _{\theta }\left( \mathtt{h}\left( 1+\gamma
\right) +\mathtt{h}_{\theta }\gamma _{\theta }\right) }{\left( 1+\gamma
\right) \left( \left( 1+\gamma \right) ^{2}+\gamma _{\theta }^{2}\right) ^{%
\frac{3}{2}}}
\end{eqnarray*}%
Finally, 
\begin{equation}
D_{\mathtt{h}}\mathbb{F}\left( \gamma ,\mathcal{\wp }_{\gamma }\right)
\left( \mathtt{h}\right) =\tau _{1}\mathtt{h}_{\theta \theta }+\tau _{2}%
\mathtt{h}_{\theta }+\tau _{3}\mathtt{h,}  \label{derivata}
\end{equation}%
where 
\begin{eqnarray*}
\tau _{1} &=&\frac{\sigma }{RP_{atm}}\frac{1+\gamma }{\left( (1+\gamma
)^{2}+\gamma _{\theta }^{2}\right) ^{\frac{3}{2}}}, \\
\tau _{2} &=&\frac{\sigma }{RP_{atm}}\frac{-3\gamma _{\theta }^{3}+3\gamma
_{\theta }\gamma _{\theta \theta }(1+\gamma )}{\left( (1+\gamma )^{2}+\gamma
_{\theta }^{2}\right) ^{\frac{3}{2}}}+\frac{\sigma }{RP_{atm}}\frac{\cot
\theta \left( \left( 1+\gamma \right) ^{2}-\gamma _{\theta }^{2}\right) }{%
\left( 1+\gamma \right) \left( \left( 1+\gamma \right) ^{2}+\gamma _{\theta
}^{2}\right) ^{\frac{3}{2}}} ,\\
\tau _{3} &=&\frac{\sigma }{RP_{atm}}\frac{-2(1+\gamma )^{3}-8\gamma
_{\theta }^{2}\left( 1+\gamma \right) }{\left( (1+\gamma )^{2}+\gamma
_{\theta }^{2}\right) ^{\frac{5}{2}}}-\frac{\sigma }{RP_{atm}}\frac{\gamma
_{\theta \theta }\gamma _{\theta }^{2}-2\gamma _{\theta \theta }(1+\gamma
)^{2}}{\left( (1+\gamma )^{2}+\gamma _{\theta }^{2}\right) ^{\frac{5}{2}}},
\\
&&+\frac{\sigma }{RP_{atm}}\cot \theta \frac{-\gamma _{\theta }^{3}+\gamma
_{\theta }\left( 1+\gamma \right) ^{2}}{\left( 1+\gamma \right) ^{2}\left(
\left( 1+\gamma \right) ^{2}+\gamma _{\theta }^{2}\right) ^{\frac{3}{2}}}%
+\psi .
\end{eqnarray*}

\begin{theorem}
\label{th1}  The Fr\'{e}chet derivative of $\mathbb{F}$ with respect to $%
\mathtt{h}$, given in \eqref{derivata}, is a linear homeomorphism. 
\end{theorem}

\begin{proof}
        Clearly $D_{\mathtt{h}}\mathbb{F}\left( \gamma ,\mathcal{\wp }_{\gamma }\right)
        $ is linear and continuous. Thus, if one can prove that it is a bijection, the Inverse Theorem guarantees that its inverse is also continuous (see, e.g., \cite{thinverse1}, \cite{thinverse2}).
        The bijectivity of $D_{\mathtt{h}}\mathbb{F}\left( \gamma ,\mathcal{\wp }_{\gamma }\right)
        $ is equivalent with proving that for any $g\in C\left[ \tfrac{\pi }{2}-\varepsilon ,\tfrac{\pi
        }{2}+\varepsilon \right]$, the initial value problem
        \begin{equation}
            \begin{cases}
                \tau _{1} u_{\theta \theta} + \tau _{2} u_{\theta} + \tau _{3} u = 0, \\
                u\left( \frac{\pi }{2}\right) = \gamma \left( \frac{\pi }{2}\right), \\
                u^{\prime }\left( \frac{\pi }{2}\right) = \gamma ^{\prime }\left( \frac{\pi }{2}\right)
            \end{cases}
        \end{equation}
        has a unique solution.
        Since $\tau _{1}\neq 0$ and $\tau_2, \tau_2$ are  continuous functions on $\left [ \frac{\pi}{2}-\varepsilon, \frac{\pi}{2}+\varepsilon\right]$,
        we employ \Cref{th initial value pb} to obtain the conclusion.

    \end{proof}
Our main result of this paper follows from \Cref{teorema functiei implicie}
and \Cref{th1}, that is,

\begin{theorem}
For any sufficiently small perturbation of $\mathcal{\wp }$ from $\mathcal{%
\wp }_\gamma$, there exists a unique $\mathtt{h} \in C_\gamma$ such that %
\eqref{ec abstracta} holds. 
\end{theorem}

\section*{Acknowledgements}


The author acknowledges the support provided by the project "Nonlinear
Studies of Stratified Oceanic and Atmospheric Flows", funded by the European
Union through the Next Generation EU initiative and the Romanian Government
under the National Recovery and Resilience Plan for Romania. The project is
contracted under number 760040/23.05.2023, cod PNRR-C9-I8-CF 185/22.11.2022,
through the Romanian Ministry of Research, Innovation, and Digitalization,
within Component 9, Investment I8. Also, the author would like to thank  Professor C.I. Martin for valuable discussions and suggestions.

\end{document}